\documentclass[conference]{IEEEtran}

\usepackage{cite}

\usepackage{graphicx}
\usepackage{amsmath,amssymb,amsfonts,amsthm}
\usepackage{bm}
\usepackage{mathrsfs}

\usepackage{booktabs}
\usepackage{multirow}
\usepackage{tabularx}
\usepackage{makecell}
\usepackage{array}
\usepackage{threeparttable}

\usepackage{caption}
\usepackage{subcaption}

\usepackage{algorithm}
\usepackage{algorithmic}
\usepackage{enumitem}

\usepackage{tikz}
\usepackage{pgfplots}
\pgfplotsset{
    compat=1.18,
    every axis plot/.append style={line width=1.2pt}
}
\usetikzlibrary{
    arrows.meta,
    fit,
    backgrounds,
    shapes.geometric
}

\usepackage{xcolor}
\usepackage{xspace}

\usepackage[disable,colorinlistoftodos]{todonotes}

\newtheorem{thm}{Theorem}

\newtheorem{ass}{Assumption}

\newcommand{\calm}{\mathcal{M}}
\newcommand{\cals}{\mathcal{S}}

\newcommand{\calt}{\mathcal{T}}

\newcommand{\cali}{\mathcal{I}}

\newcommand{\caly}{\mathcal{Y}}

\newcommand{\bx}{\boldsymbol{x}}
\newcommand{\by}{\boldsymbol{y}}

\newcommand{\bp}{\boldsymbol{p}}

\newcommand{\opt}{\texttt{OPT}\xspace}
\newcommand{\alg}{\texttt{ALG}\xspace}

\newcommand{\CR}{\texttt{CR}\xspace}

\newcommand{\osac}{\texttt{OSARA}\xspace}

\newcommand{\csp}{\texttt{CSP}\xspace}
\newcommand{\sara}{\texttt{SARA}\xspace}
\newcommand{\opa}{\texttt{OPA}\xspace}

\newcommand{\expp}{\texttt{EXP}\xspace}

\newcommand{\nr}{\texttt{NR}\xspace}

\newcommand{\mpc}{\texttt{MPC}\xspace}

\newcommand{\ie}{i.e., }

\newcommand{\fig}[1]{Fig.~\ref{#1}}

\newif\ifanonymous

\title{Data-Driven Online Slice Admission Control and Resource Allocation
in NextG Mobile Networks}

\ifanonymous

\author{\IEEEauthorblockN{Anonymous Authors}}

\else

\author{
\IEEEauthorblockN{
Muhammad Sulaiman\IEEEauthorrefmark{1},
Bo Sun\IEEEauthorrefmark{2},
Mohammad A. Salahuddin\IEEEauthorrefmark{1},
Xiaoqi Tan\IEEEauthorrefmark{3},
Raouf Boutaba\IEEEauthorrefmark{1}
}

\IEEEauthorblockA{\IEEEauthorrefmark{1}
University of Waterloo\\
\{m4sulaim, mohammad.salahuddin, rboutaba\}@uwaterloo.ca
}

\IEEEauthorblockA{
\begin{tabular}{cc}
\IEEEauthorrefmark{2} University of Ottawa
&
\IEEEauthorrefmark{3} University of Alberta
\\
bsun3@uottawa.ca
&
xiaoqi.tan@ualberta.ca
\end{tabular}
}
}

\fi

\begin{document}

\maketitle

\begingroup
\renewcommand\thefootnote{}
\footnotetext{This work has been submitted to the IEEE for possible publication.
Copyright may be transferred without notice, after which this version may no
longer be accessible.}
\addtocounter{footnote}{-1}
\endgroup

\begin{abstract}
Virtualization in 5G and beyond networks enables the creation of virtual networks (i.e., network slices) tailored to the needs of different applications. To maximize revenue under limited infrastructure resources, InPs must decide in real time whether to admit incoming slice requests (SRs) based on their resource demands and offered values, while accounting for the opportunity cost of consuming scarce resources. To address this challenge, we introduce \underline{O}nline \underline{P}ricing-based Slice \underline{A}dmission Control and Resource Allocation (\opa) framework. This framework dynamically assigns pseudo-prices to resources that capture long-term scarcity and anticipated inter-temporal opportunity costs. The short-term admission and resource allocation decisions for each SR are then guided by these prices. Additionally, we design an exponential pricing strategy that guarantees bounded worst-case performance. To improve practical performance, we further develop a data-driven exponential pricing approach that learns from historical data. Evaluations on a real-world network topology show that it improves the revenue by 32.2\% and 26.7\% over state-of-the-art DRL and optimization-based approaches, respectively, while reducing computational cost by an order of magnitude relative to the latter.
\end{abstract}

\begin{IEEEkeywords}
5G, network slicing, admission control,
resource allocation, data-driven algorithms, online algorithms
\end{IEEEkeywords}

\section{Introduction}
Network Function Virtualization (NFV) and Software-defined Networking (SDN) are enabling technologies to realize network slicing in 5G and beyond mobile networks. Network slicing  creates isolated virtual networks, atop an underlying physical infrastructure, each tailored to meet the diverse requirements of distinct service modalities, such as enhanced Mobile Broadband (eMBB) and Ultra-Reliable Low-Latency Communications (URLLC). Anticipated future developments in 5G and beyond networks envision a commoditized ecosystem \cite{TNET-1, TNSM_GNN, ONETS, vs1}, where service providers can procure network slices from infrastructure providers (InPs) to cater to specific market segments. However, constrained by limited resources, an InP may find itself unable to accommodate all slice requests (SRs). To maximize its revenue, the InP could deploy a slice admission control (SAC) mechanism that evaluates the resource demands of varying SRs against their offered revenues.

Slice Admission Control and Resource Allocation (\sara) in 5G and beyond networks presents several challenges. The first challenge arises from the \textit{online nature of the problem}. In an online \sara (\osac), SRs arrive sequentially over time, requiring admission and resource allocation decisions to be made immediately without waiting for future requests. This distinguishes \osac from offline optimization and requires algorithms that operate effectively under uncertainty. The second challenge in \sara lies in \textit{accounting for the inter-temporal impact of both admission control and resource allocation decisions}. In an end-to-end (E2E) mobile network, a SR consumes resources across multiple heterogeneous domains, including RAN, transport links, and core \cite{vs2, TNSM_GNN, sulaiman2022coordinated}. Different feasible allocations for the same SR can induce distinct resource bottlenecks,  affecting the system’s ability to accommodate future requests and making admission and allocation decisions inherently intertwined. The joint \sara problem is \textit{computationally challenging}. In practice, even when operating on small batches of SRs, state-of-the-art approaches resort to decomposition techniques and heuristic-aided optimization \cite{vs1, vs2, kansaas}. Other works simplify the problem further by assuming predetermined resource allocation per slice, thereby eliminating alternative allocation choices \cite{kansaas, Ghina_2020, vanHuynh.2019} and overlooking the critical impact of resource allocation decisions.

The final challenge concerns traditional AI-based \sara algorithms, particularly reinforcement learning (RL) approaches \cite{sulaiman2022coordinated, Ghina_2020, vanHuynh.2019, Slicepilot}. Although effective, their admission decisions are typically generated by black-box policies whose internal logic is difficult to interpret, making it unclear how factors such as resource scarcity influence admission outcomes.

\begin{figure}[t!]
 \centering
  \includegraphics[width=0.95\linewidth]{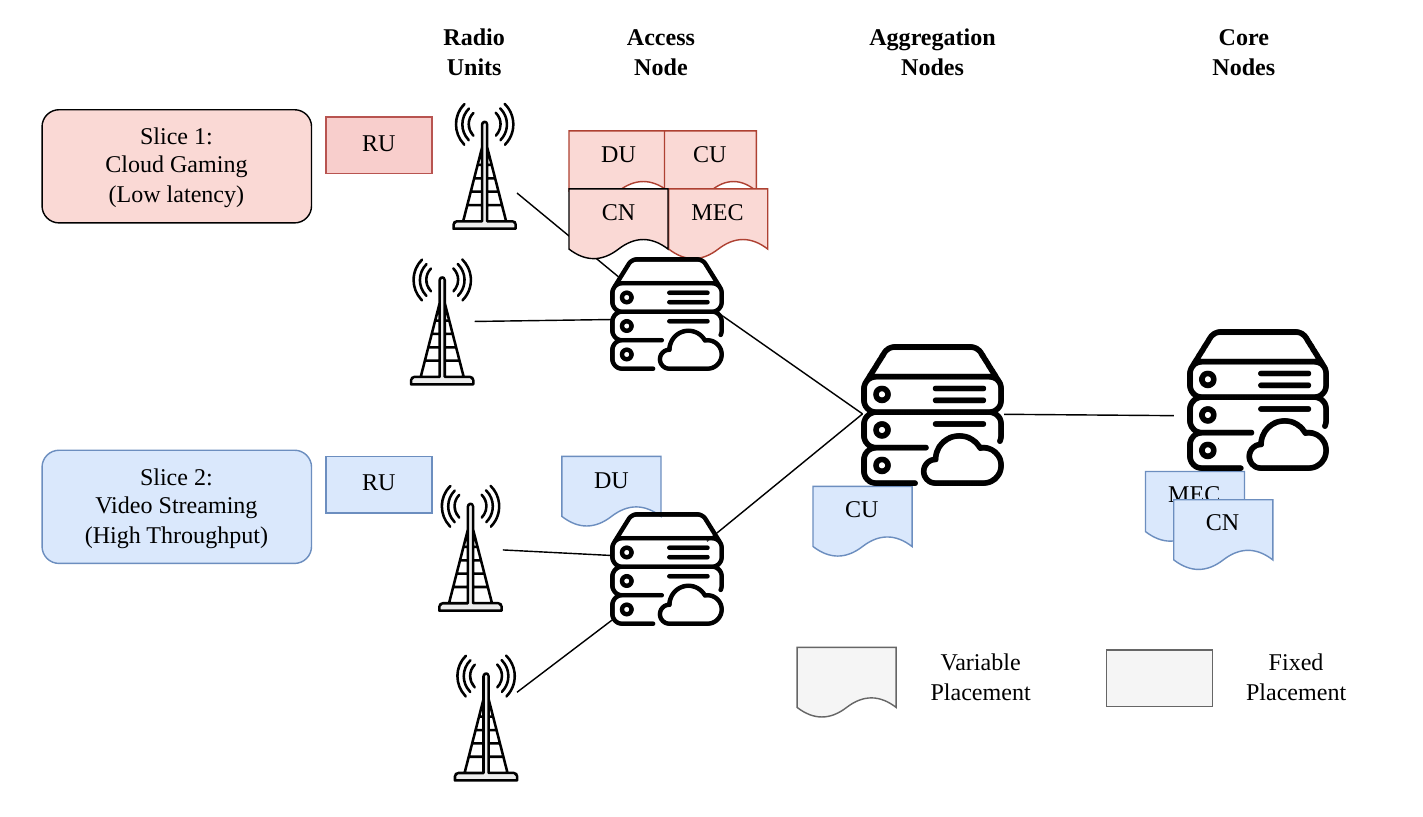}
    \caption{\small{Example VNF placements for cloud gaming (CG) and video streaming (VS) slices. CG requires CU/DU placement near the access node due to stringent latency, whereas VS allows more flexibility.}}
    \label{fig:slice_model}
\end{figure}

To address these challenges, we start by modeling resource allocation as a
virtual network embedding (VNE) problem, which is a natural abstraction for modeling
slice-level resource allocation in 5G and beyond networks. As illustrated in
Fig.~\ref{fig:slice_model}, each SR induces a virtual network composed of
interconnected virtual network functions (VNFs), \ie the Radio Unit (RU), Distributed Unit (DU), Centralized Unit (CU), Core Network (CN) and Mobile Edge Compute (MEC). These VNFs have heterogeneous resource and latency requirements and, if the SR is admitted, must be
mapped onto a shared physical substrate. To solve the joint \sara problem, we propose a novel Online Pricing-based Slice Admission Control and Resource Allocation (\opa) framework that decouples long-term admission decisions and resource pricing from short-term resource allocation, while maintaining their interaction via adaptive resource prices. Beyond proposing the pricing-based framework, we design an exponential pricing strategy that guarantees worst-case performance even in adversarial settings. To improve empirical performance, we further develop an adaptive exponential pricing approach that learns the pricing parameters through data. Our main contributions are:


\begin{itemize}[leftmargin=1.5em,labelsep=0.45em]
    \item We formulate \textbf{\osac as a joint problem} to explicitly capture their interdependence. Unlike prior work that assumes predetermined resource demands, our formulation treats resource allocation as a decision variable. 

    \item We propose \textbf{a novel Online Pricing-based \sara frame-work} that decouples long-term admission control from short-term resource allocation through resource pseudo-prices. This reduces computational complexity while retaining the ability to account for long-term system impact.

    \item We develop \textbf{a simple yet effective exponential pricing strategy for \osac} that is proven to attain a bounded worst-case performance under the competitive analysis framework \cite{borodin2005online}.
    Building on this design, we develop an adaptive variant that leverages offline training to enhance empirical performance. 

    \item Unlike existing \sara approaches that rely on black-box policies, our framework enables \textbf{ interpretability} through utilization dependent resource pricing, allowing InPs to analyze admission decisions via cost-revenue signals.
    
\end{itemize}

\section{Related Work}\label{sec:bg_related}

\noindent\textbf{Resource allocation.} Prior work on RAN-only resource allocation employs RL, Bayesian optimization, conventional optimization, and model-driven methods~\cite{liu2021constraint,noms23,liu2021onslicing,liu2022atlas,kasgari2018stochastic,microopt,zipper}. Although effective in the RAN, these approaches do not address E2E orchestration, where heterogeneous resources must be allocated across multiple network segments with seconds to minutes of optimization timescales. For E2E allocation, \cite{koo2019deep} considers multidimensional resources and delay constraints but abstracts the slice and substrate as aggregated nodes. RAN-function placement and routing is addressed in \cite{gao2021deep}, while \cite{vs1} jointly optimizes access control, CU selection, and resource reservation. Neither models per-VNF delay constraints. Although \cite{solozabal2019virtual} considers delay-aware VNF placement, it assumes a star substrate topology. Joint admission and allocation are also studied using DRL and MARL~\cite{vanHuynh.2019,RR_2019,sulaiman2022coordinated,noms22,TNSM_GNN,Slicepilot}, with \cite{TNSM_GNN} using GNNs to improve topology generalization. SlicePilot~\cite{Slicepilot} places VNFs across heterogeneous cloud tiers but assumes aggregate slice latency and a linear substrate. Overall, prior work often simplifies the substrate or omits per-VNF resource and delay constraints, limiting its ability to capture joint E2E placement and routing.\\

\noindent\textbf{Slice admission control.} Most 5G SAC approaches use RL \cite{sulaiman2022coordinated,vanHuynh.2019,RR_2019,TNSM_GNN,Roig.2019,bega2019machine} or multi-armed bandits \cite{ONETS}. RL-based admission and congestion control is studied in \cite{Ghina_2020}, while \cite{vanHuynh.2019} jointly considered SAC and resource allocation. \cite{bega2019machine} used separate agents to estimate the revenue from accepting or rejecting slice requests. Multi-agent DRL (MARL) was also used for joint SAC and resource allocation in \cite{sulaiman2022coordinated,noms22,TNSM_GNN}, with \cite{TNSM_GNN} incorporating GNNs to support large, dynamic substrate topologies. Outside RL, \cite{ONETS} modeled online SAC as a budgeted lock-up bandit. However, it only optimizes per-round rewards instead of over a horizon. Competitive markets with multiple infrastructure providers and rational tenants were considered in \cite{TNET-1}, but also assumed fixed resource-allocation and did not jointly optimize allocation. On the other hand, yield-driven overbooking methods \cite{vs1,vs2,kansaas} periodically solve batched SARA problems using traffic forecasts. These model-predictive-control approaches balance overbooking revenue against SLA-violation risk over a finite prediction horizon.
\section{Problem Statement} \label{sec:prob_statement}
\subsection{\underline{S}lice \underline{A}dmission Control and \underline{R}esource \underline{A}llocation (\sara)} \label{sec:sara}

We address the \sara problem for 5G and beyond networks. Consider a substrate network with $M$ resources, where each resource $m\in \calm:=\{1,\dots,M\}$ represents a generic resource, e.g., the bandwidth of a link or the computing resource of a node. Let $C_m$ denote the capacity of resource $m$.  
We consider a time-slotted system with a slot set $\calt:=\{1,\dots,T\}$. A sequence of SRs  
$\cals :=\{1,\dots,S\}$ arrive over time. The information of each SR $i$ is represented by $I_i:= \{v_i, \boldsymbol{\xi}_i, \calt_i\}$. $v_i$ is the value of SR, indicating the price the service provider is willing to pay for the request. $\calt_i:= \{t^a_i,\dots, t^d_i\}$ is the service period, where $t^a_i$ and $t^d_i$ are the arrival and departure slots of the SR, respectively. 
$\boldsymbol{\xi}_i$ is a slice feature vector that encodes all the necessary information to determine the set of feasible resource allocations. 
Formally, feasibility is defined by a set of constraints $\mathcal{C}_i(\boldsymbol{\xi}_i)$ induced by the slice request and the substrate network. For example, a transport network bandwidth requirement may be realized by allocating capacity across multiple transport network links subject to capacity, flow-conservation and delay constraints. We denote by $\mathcal{F}_i(\boldsymbol{\xi}_i)$ the set of resource allocations that satisfy $\mathcal{C}_i(\boldsymbol{\xi}_i)$.
Details of $\mathcal{C}_i(\boldsymbol{\xi}_i)$ and $\mathcal{F}_i(\boldsymbol{\xi}_i)$ are provided in Section~\ref{sec:vne}.

\noindent \textbf{Offline problem.} 
For each SR $i\in\cals$, the problem is to determine the admission control decision $x_i \in \{0,1\}$, i.e., whether to admit SR $i$, and a corresponding resource allocation $\by_i := \{y_{i,m}^t\}_{m\in\calm,t\in\calt}$. $y_{i,m}^t$ is the amount of resource $m$ allocated to SR $i$ at time slot $t$. Given $x_i$, let $\caly_i(x_i)$ denote the feasible set of resource allocation $\by_i$. 
If SR $i$ is rejected (i.e., $x_i = 0$), no resource is allocated, and then $\caly_i(0):= \{\by_i: y_{i,m}^t = 0,\forall m\in\calm, t\in\calt\}$.
If SR $i$ is admitted (i.e., $x_i = 1$), a feasible resource allocation must guarantee the feasibility requirement, therefore, $\caly_i(1) := \{\by_i : \by_i \in \mathcal{F}_i(\boldsymbol{\xi}_i)\}$. 

The goal of \sara is to maximize the total value of admitted slice requests, while ensuring that all selected allocations are feasible and that the capacities of all resources are respected. Let $\cali :=\{v_i, \boldsymbol{\xi}_i, \calt_i\}_{i\in\cals}$ denote an instance of the problem.
Given the information $\cali$ of all slices, the offline problem is:
\begin{subequations}
\label{p:sac}
\begin{align}
\label{eq:sac-obj}
    \max_{x_i, \by_{i}} \quad& \sum\nolimits_{i\in\cals} v_i x_i\\
    \label{eq:sac-capacity}
    {\rm s.t.}\quad & \sum\nolimits_{i\in\cals} y_{i,m}^t \le C_m, \forall m\in \calm, t\in\calt,\\ 
    \label{eq:sac-qos}
    &x_{i} \in \{0,1\}, \by_i \in \caly_i(x_i), \forall i\in\cals,
\end{align}
\end{subequations}

\subsection{Resource Allocation as \underline{V}irtual \underline{N}etwork \underline{E}mbedding (VNE)} \label{sec:vne}
In this subsection, 
we formally define the constraints $\mathcal{C}_i(\boldsymbol{\xi}_i)$ and feasible solutions
$\mathcal{F}_i(\boldsymbol{\xi}_i)$ for resource allocation.
In particular, we model slice-level resource allocation as a VNE problem on a substrate network.

\noindent \textbf{Substrate network.} Metropolitan 5G networks follow a hierarchical multi-tier structure with access, aggregation, and core nodes \cite{NGMN_topo, ituseriesG8300, metrohaul} (cf. Fig. \ref{fig:topos}). Higher tier nodes provide greater compute capacity but incur additional transport latency. Access nodes connect to aggregation nodes via Tier-1 links, and aggregation nodes connect to core nodes via Tier-2 links. We model this metropolitan 5G substrate network as an undirected graph $G=(\mathcal V,\mathcal E)$, where $\mathcal V$ denotes the set of substrate nodes and $\mathcal E$ denotes the set of substrate links connecting them. Node $v\in\mathcal V$ has CPU and memory capacities $C^{\text{cpu}}_v$ and $C^{\text{mem}}_v$, respectively. Link $e\in\mathcal E$ has bandwidth capacity $B_e$ and propagation delay $d_e$.

\begin{figure}[h!]
 \centering
  \includegraphics[width=\linewidth]{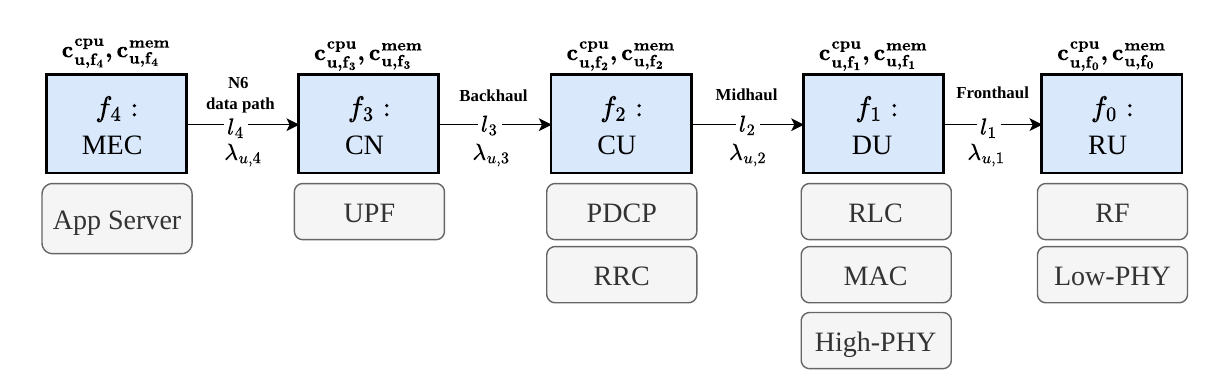}
    \caption{Virtual network induced by SR $i$.}
    \label{fig:VNF_model}
\end{figure}

\noindent \textbf{Slice requests.} 
We denote by $\mathcal{U}$ the set of slice types.  
Each type $u\in\mathcal U$ is modeled as a virtual network (VN) as illustrated in Fig.~\ref{fig:VNF_model}, which consists of a chain of five VNFs in $\mathcal F=\{f_0,\dots,f_4\}$, interconnected by four
virtual links (VLs) $\mathcal L=\{l_1,\dots,l_4\}$. Aligned with 3GPP and \mbox{ITU-T} recommendation \cite{ituseriesG66}, we consider split 7.2. For slice type $u$, each VNF
$f\in\mathcal F$ requires $c^{\text{cpu}}_{u,f}$ CPU and $c^{\text{mem}}_{u,f}$ memory,
while each VL $l\in\mathcal L$ requires bandwidth $\lambda_{u,l}$. In addition, slice
type $u$ specifies delay constraints $D_{u,f}$ up to VNF $f$
(e.g., RU$\rightarrow$DU, RU$\rightarrow$CU). Recall that SRs are specified by
$I_i:=\{v_i,\mathcal T_i,\boldsymbol{\xi}_i\}$. The slice feature vector $\boldsymbol{\xi}_i$ contains the complete specification
of the VN induced by SR $i$, including its slice type $u_i$,
originating access node $n_i^{\text{src}} \in \mathcal V$, per-VNF and per-VL
resource demands ($\{c^{\text{cpu}}_{u_i,f},c^{\text{mem}}_{u_i,f}\}_{f\in\mathcal F}$,
$\{\lambda_{u_i,l}\}_{l\in\mathcal L}$) and the delay constraints $\{D_{u_i,f}\}_{f\in\mathcal{F}}$.

\noindent \textbf{Embedding decision variables.} If SR $i$ is admitted, its VN must be embedded onto the substrate network by
: (i) placing each VNF on a substrate node, and (ii) routing each VL
along a single substrate path. Let $q \in \mathcal Q$ denote a simple path in the substrate network, where a path is
defined as a sequence of substrate links.
Let $\delta^i_{f,v}\in\{0,1\}$ indicate whether VNF $f$ of SR $i$ is placed on
substrate node $v$, and let $\rho^i_{l,q}\in\{0,1\}$ indicate if virtual link $l$ of
SR $i$ is routed along path $q\in\mathcal Q$. 

\noindent \textbf{Resource allocation constraints $\mathcal{C}_i(\boldsymbol{\xi}_i)$.}

\noindent (i) \textit{Admission and embedding constraints}:
If SR $i$ is admitted (i.e., $x_i =1$), each of its VNFs must be placed on exactly one eligible
substrate node, and each of its VLs must be routed on exactly one substrate path. This is
enforced by:
\begin{align}
x_i &= \sum_{v \in \mathcal V_{i,f}} \delta^i_{f,v},\ \forall f \in \mathcal F,
\qquad
x_i = \sum_{q \in \mathcal Q} \rho^i_{l,q},\ \forall l \in \mathcal L,
\label{eq:vne_adm}
\end{align}
where $\mathcal V_{i,f} \subseteq \mathcal V$ is the set of eligible substrate nodes
for VNF $f$. The RU of SR $i$ is fixed at its originating access node,
i.e., $\mathcal V_{i,f_0}=\{n^{\text{src}}_i\}$. 

\noindent (ii) \textit{Routing (flow-conservation) constraints}:
Let $\mathrm{src}(l), \mathrm{dst}(l)\in\mathcal F$ denote the source and destination VNFs
of virtual link $l$. We adopt single-path routing for each VL, which can be modeled as an
unsplittable multi-commodity flow. For each SR $i$, virtual link $l$, and
substrate node $v$, flow conservation is enforced as:
\begin{equation}
\sum_{\substack{q\in\mathcal Q:\\ \mathrm{src}(q)=v}} \rho^i_{l,q}
-
\sum_{\substack{q\in\mathcal Q:\\ \mathrm{dst}(q)=v}} \rho^i_{l,q}
=
\delta^i_{\mathrm{src}(l),v}
-
\delta^i_{\mathrm{dst}(l),v},
\ \forall v\in\mathcal V,\; l\in\mathcal L,
\label{eq:vne_flow}
\end{equation}
where $\mathrm{src}(q), \mathrm{dst}(q) \in \mathcal V$ denote the source and destination nodes of path $q$. We consider that $\emptyset \in \mathcal Q$ with
$\mathrm{src}(\emptyset)=\mathrm{dst}(\emptyset)$ to model co-location of VNFs.

\noindent (iii) \textit{Capacity constraints}:
For a given SR $i$, the embedding must respect the resource
capacities of nodes and links, \ie
\begin{equation}
\begin{array}{@{}l@{\;}c@{\;}l@{\quad}l@{}}
\displaystyle
\sum_{f\in\mathcal F} c^r_{u_i,f}\delta^i_{f,v}
&\le&
C^r_v,
& r\in\{\mathrm{cpu},\mathrm{mem}\},\ \forall v\in\mathcal V,
\\[2mm]
\displaystyle
\sum_{l\in\mathcal L}
\sum_{\substack{q\in\mathcal Q\\ e\in q}}
\lambda_{u_i,l}\rho^i_{l,q}
&\le&
B_e,
& \forall e\in\mathcal E.
\end{array}
\label{eq:vne_capacity}
\end{equation}

Per-SR constraints in~\eqref{eq:vne_capacity} combined with the system-level capacity constraint \eqref{eq:sac-capacity}, which aggregates allocations across all admitted slices and time slots, guarantee that resource usage remains feasible with respect to residual substrate capacities throughout the time horizon.

\noindent (iv) \textit{Delay constraints}: For each SR $i$, the E2E delay accumulated along the substrate paths
supporting the virtual links up to any intermediate VNF must not exceed the corresponding
delay bound. Formally,
\begin{equation}
\sum_{\substack{l\in\mathcal L:\\ l \le f}}
\sum_{q\in\mathcal Q}
\rho^i_{l,q}
\left(\sum_{e\in q} d_e\right)
\le D_{u_i,f},
\quad \forall i,\; f \in \{1,2,3,4\}.
\label{eq:vne_delay}
\end{equation}

\noindent \textbf{Feasible resource allocation set $\mathcal{F}_i(\boldsymbol{\xi}_i)$.} For a SR $i$, the set $\mathcal{F}_i(\boldsymbol{\xi}_i)$ consists of all
virtual network embeddings $\{\delta^i_{f,v}, \rho^i_{l,q}\}$ that satisfy
constraints (2)--(5) when the slice is admitted
(i.e., $x_i=1$). Each feasible embedding induces a corresponding abstract resource allocation
$\by_i$ over the service period $\mathcal T_i$, and vice versa. Specifically, let each abstract resource index $m\in\mathcal M$ correspond to either:
(i) a node resource $(v,\mathrm{cpu})$, $(v,\mathrm{mem})$, or
(ii) a link resource $(e,\mathrm{bw})$.
Then, for any admitted SR $i$ and time slot $t\in\mathcal T_i$,
the induced abstract allocation $y_{i,m}^t$ is defined as follows:

\begin{equation}
\begin{aligned}
y_{i,(v,r)}^t
 &= \sum_{f\in\mathcal F} c^r_{u_i,f}\delta^i_{f,v},
 && r\in\{\mathrm{cpu},\mathrm{mem}\},\\
y_{i,(e,\mathrm{bw})}^t
 &= \sum_{l\in\mathcal L}
    \sum_{\substack{q\in\mathcal Q\\ e\in q}}
    \lambda_{u_i,l}\rho^i_{l,q}.
\end{aligned}
\end{equation}

\subsection{Online Slice Admission Control and Resource Allocation}
We also formulate an online version of \sara (\osac),
where the set of SRs arrive one by one. For each arrival, we must make irrevocable admission and resource allocation decisions immediately without the information of future SRs. Unlike some prior works that assume batch processing of SRs \cite{vs1, vs2, kansaas}, our formulation addresses the fully online setting with per-arrival decisions. The batch model is therefore a restricted instance of our framework.

Let $\alg(\cali)$ and $\opt(\cali)$ denote the total values obtained by an online algorithm and the offline algorithm, respectively, on instance $\cali$. The performance of the online algorithm is evaluated by its competitive ratio ($\CR$), i.e., $\CR = \max_{\cali \in \Omega} {\opt(\cali)}/{\alg(\cali)}$, where $\Omega$ is the set of all possible instances. $\CR$ is a classic information-theoretic performance metric, which quantifies the performance of an online algorithm versus the offline algorithm in the worst-case scenario under the framework of competitive analysis~\cite{borodin2005online}. 
An algorithm with bounded \CR ensures robustness, and we aim at designing an online algorithm that minimizes $\CR$.
\section{\underline{O}nline \underline{P}ricing-based Slice \underline{A}dmission Control and Resource Allocation (\opa)}
\label{sec:solution}

We propose $\opa$, an online price-based framework for \osac. The core idea is to estimate the cost of serving each slice based on the current state of the substrate network, and admit the slice only if its value (i.e., revenue) exceeds this cost. To this end, 
we introduce a pricing function 
$\phi(\boldsymbol{z})$
that maps the current network state $\boldsymbol{z}$ (e.g., resource utilization) to resource prices. For each arriving slice $i$, the unit price of resource $m$ at time $t$ is given by $p_{m,t}^{(i)} = \phi(\boldsymbol{z})$. 
Different pricing rules can be instantiated within this framework, including fixed pricing, utilization-based pricing, dual-driven pricing, or learning-based pricing schemes. 
Given the resulting price vector $\bp^{(i)}$, $\opa$ estimates the cost of slice $i$ by solving the following single-slice cost-minimization problem.

\noindent \textbf{\underline{C}ost-minimization \underline{S}ingle-slice \underline{P}roblem (\csp).}
For each slice $i\in\cals$, given the price vector $\bp^{(i)}$ and slice information
$\{\boldsymbol{\xi}_i,\calt_i\}$, $\opa$ solves $\csp(\bp^{(i)};\boldsymbol{\xi}_i,\calt_i)$
to find a feasible allocation that minimizes the cost of serving slice $i$:
\begin{subequations}
\label{p:csp}
\begin{align}
\min_{\by_i} \quad 
    & \sum_{t\in\calt_i}\sum_{m\in\calm} p_{m,t}^{(i)}\, y_{i,m}^t \\
{\rm s.t.}\quad 
    & \by_i \in \mathcal{F}_i(\boldsymbol{\xi}_i), \\
    & 0 \le y_{i,m}^t \le R_m, \quad \forall m\in\calm,\ t\in\calt_i,
\end{align}
\end{subequations}
where $R_m$ denotes the maximum allocatable capacity of resource $m$ (with $R_m \le C_m$). The resource allocation of each slice $i$ for resource $m$ is constrained by $R_{m}$ to prohibit the slice from exhausting resource $m$. Let $\tilde{\by}_i$ denote a feasible solution of \csp, with corresponding cost
$\tilde{c}_i$. Then the admission control only admits slice $i$ if the slice's value is larger than the cost $\tilde{c}_i$. When admitted, the allocation $\tilde{\by}_i$ is used.

The effectiveness of $\opa$ depends on the design of the pricing function $\phi(z)$. The pricing function must address three fundamental challenges for effective slice admission control and resource allocation within this framework. First, in the absence of knowledge about future SR arrivals, it must balance conservative and greedy behaviors by accounting for the opportunity cost of resource consumption, preventing premature exhaustion of scarce resources while avoiding under-utilization. Second, it should strategically differentiate resources according to their structural importance within the network. For example, low-delay links may be inherently more valuable because they can accommodate slices with stricter QoS requirements, and should therefore be priced accordingly. Finally, since \csp may be solved approximately, the pricing function must be robust to this sub-optimality. In the following, we describe our approach to solving \csp, followed by our proposed pricing strategy in Section~\ref{sec:expp}.
\smallskip

\noindent \textbf{Practical Solvability for \csp.}
\label{sec:csp_practical}
Although smaller than the full joint admission-and-allocation problem, the single-slice cost-minimization problem (\csp) remains NP-hard because it jointly optimizes discrete placement and routing under node-capacity, link-capacity, and cumulative-delay constraints. It subsumes NP-hard problems such as delay-constrained unsplittable multi-commodity routing \cite{mcp_hard}. However, \csp\ removes admission variables and inter-slice coupling, so its size depends only on the substrate and one slice request. Slice-specific resource and delay constraints further prune infeasible nodes and links, substantially reducing the candidate space. Therefore, modern MIP solvers such as Gurobi can obtain high-quality feasible embeddings within short runtimes, as shown in Section~\ref{sec:results}, making \csp\ practical for online deployment.

Inspired by \cite{kestrel}, we adopt a two-tier strategy that combines time-limited optimization with a lightweight heuristic fallback. The InP specifies a per-slice solver time limit, i.e.,  if optimality is not proven within this limit, the best feasible solution found is used. If no feasible solution is available, the node-ranking heuristic in Section~\ref{sec:node-ranking} greedily constructs an embedding based on residual capacities and delay constraints. This strategy provides predictable decision latency and remains robust to rare difficult instances.

\section{\underline{EX}ponential \underline{P}ricing (\expp) with Worst-Case Guarantees} \label{sec:expp}

We introduce our pricing strategy focusing on worst-case performance guarantees under the competitive analysis framework~\cite{borodin2005online}. $\expp$ is designed 
to provide robustness against adversarial or highly uncertain SR sequences. For this, $\expp$ requires the following assumptions.

\begin{ass}\label{ass1}
For each slice $i\in\cals$, its value $v_i$ and any feasible resource allocation $\by_i \in \caly_i(1)$ satisfy conditions:

\begin{enumerate}[label=(\roman*)]
\item value density of each slice is bounded, i.e., 
\begin{align} \label{eq:ass-1}
    \frac{v_i}{T_i \sum_{m\in\calm}y_{i,m}^t} \in [L,U], \forall t\in \calt_i,
\end{align}
where $T_i : = |\calt_i|$ is the length of the stay duration.

\item variation of resource allocation is upper bounded, i.e.,
\begin{align}\label{eq:ass-2}
    \frac{\sum_{m\in\calm}y_{i,m}^t}{\min_{m\in\calm: y_{i,m}^t>0}y_{i,m}^t} \le V, \forall t \in\calt_i.
\end{align}

\item resource allocation is small compared to capacity, i.e.,
\begin{align}\label{eq:ass-3}
    \max_{i\in\cals,t\in\calt_i} y_{i,m}^t < \frac{C_m}{1+ \log_2(\sigma UKV/L + 1)}, \forall m\in\calm,
\end{align}
where $K = \max_{i\in\cals} T_i$ is the maximum stay duration.

\item $\sigma$-approximate solvability of \csp:
There exists an algorithm for \csp\ such that for each slice $i$,
\begin{align}
    \tilde{c}_i \le \sigma c_i^*,
\end{align}
where $c_i^*$ denotes the optimal cost of \csp\ and $\tilde{c}_i$ is the cost returned by a sub-optimal solver, with $\sigma \ge 1$.
\end{enumerate}
\end{ass}
 Condition~\eqref{eq:ass-1} requires that the value of each slice is proportional to its total resource consumption in each slot and stay duration, and the factor of the proportionality is uncertain but within bound $[L, U]$. Condition~\eqref{eq:ass-2} means that the resource allocation over different resources has a maximum variation of $V$. Condition~\eqref{eq:ass-3} assumes that the consumption of one slice for any resource is small compared to the capacity, which is reasonable in practice.    
The parameters $L, U, V$ and $K$ 
can be determined by InP using SR templates.
The $\expp$ pricing function is defined as 
$\phi:=\{\phi_{m,t}\}_{m\in\calm,t\in\calt}$ with
\begin{align}
\label{eq:threshold}
\phi_{m,t}(w)
= L\left[\exp\!\left(\frac{\alpha w}{2C_m}\right) - 1\right],
\quad w\in[0,C_m],
\end{align}
where $C_m$ denotes the capacity of resource $m$, and 
$w$ represents the current utilization of resource $m$ at time $t$, i.e.,
$w = w_{m,t}^{(i-1)}$, the total amount of resource $m$ allocated at time $t$
after processing the first $i-1$ SRs. The parameter $\alpha$ is defined as $\alpha = 2\ln\!\left(\frac{\sigma UVK}{L} + 1\right) + 2\ln 2.$
Since \csp\ may only be solved approximately, \expp\ adopts a slightly modified admission rule. In particular, the admission control admits slice $i$ if $v_i \ge \tilde{c}_i / \sigma,$
where $\tilde{c}_i$ is the cost returned by the \mbox{$\sigma$-approximate} solver. This scaling compensates for solver sub-optimality and is required to establish the following worst-case guarantee:

\begin{thm}\label{thm:osac}
Under Assumption~\ref{ass1},
$\opa$ instantiated with $\expp$ is 
$\frac{(\sigma + 1)\alpha}{2\ln 2}$-competitive
for \osac. If \csp is solved optimally
(i.e., $\sigma = 1$), $\opa$ achieves a competitive ratio of $\alpha$.
\end{thm}

\noindent \textbf{Proof of Theorem~\ref{thm:osac}.}
We analyze the competitive performance of \opa based on the online primal-dual analysis approach~\cite{buchbinder2009design}.
By partially relaxing the capacity constraint~\eqref{eq:sac-capacity} using the dual variable $\boldsymbol{\lambda}:=\{\lambda_{m,t}\}_{m\in\calm,t\in\calt}$, the dual of problem~\eqref{p:sac} is
\begin{align*}
\min_{\boldsymbol{\lambda} \ge 0}\max_{\substack{x_i \in \{0,1\},\\ \by_{i}\in\caly_{i}(x_i)}} \quad \sum_{i\in\cals} v_i x_i + \sum_{t\in\calt}\sum_{m\in\calm} \lambda_{m,t} [C_m -  \sum_{i\in\cals} y_{i,m}^t].
\end{align*}
Equivalently, the dual problem can be presented as:
\begin{align*}
\min_{\boldsymbol{\lambda}\ge 0} \quad&  \sum_{t\in\calt}\sum_{m\in\calm} \lambda_{m,t} C_m + \sum_{i\in\cals} \hat{x}_i\left[v_i - \sum_{t\in\calt}\sum_{m\in\calm} \hat{y}_{i,m}^t \lambda_{m,t} \right], 
\end{align*}
where for a given $\boldsymbol{\lambda}$, $\hat{\by}_i:= \hat{\by}_i(\boldsymbol{\lambda}) = \{\hat{y}_{i,m}^t\}_{m\in\calm, t\in\calt}$ is the optimal solution of 
$$\min_{\by_i \in \caly_i(1)} \sum\nolimits_{t\in\calt_i}\sum\nolimits_{m\in\calm} \lambda_{m,t} y_{i,m}^t,$$
and $\hat{\bx}:=\{\hat{x}_i\}_{i\in\cals}$ is the optimal solution of 
$$\max_{x_i\in\{0,1\}} \sum_{i\in\cals} {x}_i\left[v_i - \sum_{t\in\calt}\sum_{m\in\calm} \hat{y}_{i,m}^t \lambda_{m,t} \right].$$ 
Therefore, $\hat{\by}_i$ is the optimal solution of $\csp(\boldsymbol{\lambda};\boldsymbol{\xi}_i,\calt_i)$, and
$\hat{x}_i = 1$ if $v_i \ge \sum_{t\in\calt}\sum_{m\in\calm} \hat{y}_{i,m}^t \lambda_{m,t}$ and $\hat{x}_i = 0$ otherwise.

The high-level idea of online primal-dual analysis is to construct a feasible dual solution of problem~\eqref{p:sac} based on the solution from the online algorithm \opa. Let $\texttt{Dual}(\cali)$ denote the dual objective evaluated at the feasible solution, then \opa is $\alpha$-competitive if the following inequality holds: 
\begin{align}\label{eq:opd}
    \alpha \cdot \alg(\cali) \ge  \texttt{Dual}(\cali) \ge  \opt(\cali). 
\end{align}

First, we show the second inequality, which holds based on weak duality when the solution $\{\bar{x}_i, \bar{\by}_i\}_{i\in\cals}$ of \opa is primal feasible and the constructed dual solution is dual feasible.
We note that the online decision of \opa satisfies constraint~\eqref{eq:sac-qos} directly. 
Therefore, we just need to show that no resource capacity constraints can be violated by the online decision of \opa. To see this, suppose resource $m'$ reaches the capacity, then for any follow-up slice that uses $m'$, the scaled estimated cost of admitting this slice is at least
\begin{align*}
\Tilde{c}_i/\sigma &\ge 
   \tilde{y}_{i,m'}^t \phi_{m',t}(C_{m'})/\sigma\\ &= \tilde{y}_{i,m'}^t VUK
   \ge \sum\nolimits_{m\in\calm}y_{i,m}^t U T_i \ge v_i, 
\end{align*}
where the first and the second inequalities hold due to conditions~\eqref{eq:ass-1} and~\eqref{eq:ass-2} in Assumption~\ref{ass1}.
Therefore, the follow-up slice will not be admitted by \opa and no capacity constraints will be violated.

Given the online decision of \opa, we construct a solution of the dual problem as: 
\begin{align}
    \bar{\lambda}_{m,t} = \phi_{m,t}(w^{(S)}_{m,t}), \forall m\in\calm, t\in\calt,
\end{align}
where $w^{(S)}_{m,t} = \sum_{i\in\cals} \bar{y}_{i,m}^t$ is the final utilization of resource $m$ at time $t$ when running \opa.
Clearly, $\bar{\lambda}_{m,t} \ge 0, \forall m\in\calm, t\in\calt$ and thus the dual solution is feasible.

Next, we show the first inequality in~\eqref{eq:opd}. Let $P_i$ and $D_i$ denote the primal and dual objective after processing the $i$-th slice using \opa.
The increment of the primal objective is 
\begin{align*}
    P_i - P_{i-1} = v_i \bar{x}_i,
\end{align*}
and the increment of the dual objective is 
\begin{align*}
    D_i - D_{i-1} &= \sum_{t\in\calt}\sum_{m\in\calm} [\phi_{m,t}(w^{(i)}_{m,t}) - \phi_{m,t}(w^{(i-1)}_{m,t})]C_m   \\
    &\quad\quad+\hat{x}_i [v_i - \sum_{t\in\calt}\sum_{m\in\calm} \hat{y}_{i,m}^t \phi_{m,t}(w^{(S)}_{m,t}) ].
\end{align*}

To relate the dual increment $D_i - D_{i-1}$ with online decision $\bar{x}_i$ and $ \bar{\by}_i$ of \opa, note that
\begin{subequations}
\begin{align}
  \hat{x}_i[v_i - \sum_{t\in\calt}&\sum_{m\in\calm} \hat{y}_{i,m}^t \phi(w^{(S)}_{m,t}) ] \\
  &\le \hat{x}_i [v_i - \sum_{t\in\calt}\sum_{m\in\calm} \hat{y}_{i,m}^t \phi(w^{(i-1)}_{m,t}) ]  \\
  &\le \hat{x}_i [v_i - \frac{1}{\sigma}\sum_{t\in\calt}\sum_{m\in\calm} \bar{y}_{i,m}^t \phi(w^{(i-1)}_{m,t}) ]\\
  &\le \bar{x}_i [v_i - \frac{1}{\sigma}\sum_{t\in\calt}\sum_{m\in\calm} \bar{y}_{i,m}^t \phi(w^{(i-1)}_{m,t}) ],
\end{align}    
\end{subequations}
where the first inequality holds since $\phi_{m,t}(\cdot)$ is a non-decreas-ing function.
The second inequality holds because $\hat{\by}_i$ and $\bar{\by}_i$ are, respectively, the feasible solution and $\sigma$-approximate solution of  $\csp(\bp^{(i)};\boldsymbol{\xi}_i,\calt_i)$. Therefore, we have $$\sum_{t\in\calt}\sum_{m\in\calm} \hat{y}_{i,m}^t \phi(w^{(i-1)}_{m,t}) \ge c_i^* \ge \frac{1}{\sigma}\sum_{t\in\calt}\sum_{m\in\calm} \bar{y}_{i,m}^t \phi(w^{(i-1)}_{m,t}).$$ 
The last inequality holds because $\hat{x}_i = 1$ must give $\bar{x}_i = 1$.
Next we can consider the following two cases.

\noindent\textbf{Case I.} When $\bar{x}_i = 0$, we have $\bar{y}_{i,m}^t = 0, \forall m\in\calm, t\in\calt$,
and thus $D_i - D_{i-1} \le 0 = P_{i} - P_{i-1}. $

\noindent\textbf{Case II.} When $\bar{x}_i = 1$, we have $\bar{\by}_i = \tilde{\by}_i$, and the admission of slice $i$ gives
$v_i \ge \frac{1}{\sigma}\sum_{t\in\calt_i}\sum_{m\in\calm} \bar{y}_{i,m}^t \phi_{m,t}(w^{(i-1)}_{m,t})$. Thus,
\begin{subequations}
\begin{align}
    D_i - D_{i-1} &\le \sum_{t\in\calt}\sum_{m\in\calm} [\phi_{m,t}(w^{(i)}_{m,t}) - \phi_{m,t}(w^{(i-1)}_{m,t})]C_m \nonumber\\
    &\quad\quad+  v_i - \frac{1}{\sigma}\sum_{t\in\calt}\sum_{m\in\calm} \bar{y}_{i,m}^t \phi_{m,t}(w^{(i-1)}_{m,t}) \\
    &\approx (\frac{\alpha}{2} - \frac{1}{\sigma})\sum_{t\in\calt_i}\sum_{m\in\calm}\bar{y}_{i,m}^t\phi_{m,t}(w_{m,t}^{(i-1)}) \nonumber\\
    \label{eq:ineq1}
    &\quad\quad\quad + \frac{\alpha}{2} \sum_{t\in\calt_i}\sum_{m\in\calm}\bar{y}_{i,m}^t L + v_i \\
    \label{eq:ineq2}
    &\le \frac{(\sigma+1)\alpha}{2} v_i  = \frac{(\sigma+1)\alpha}{2} (P_i - P_{i-1}).
\end{align}    
\end{subequations}
The equality~\eqref{eq:ineq1} holds since
\begin{subequations}
\begin{align*}
&\sum_{t\in\calt}\sum_{m\in\calm} [\phi_{m,t}(w^{(i)}_{m,t}) - \phi_{m,t}(w^{(i-1)}_{m,t})]C_m \\
&= C_m\sum_{t\in\calt_i}\sum_{m\in\calm} L \exp(\frac{\alpha w_{m,t}^{(i-1)}}{2 C_m}) \cdot [\exp(\frac{\alpha \bar{y}_{i,m}^t}{2 C_m} ) - 1] \\
&\approx C_m\sum_{t\in\calt_i}\sum_{m\in\calm} L \exp(\frac{\alpha w_{m,t}^{(i-1)}}{2 C_m} ) \cdot \frac{\alpha \bar{y}_{i,m}^t}{2 C_m} \\
& = \frac{\alpha}{2} \sum_{t\in\calt_i}\sum_{m\in\calm}\bar{y}_{i,m}^t[L \exp(\frac{\alpha w_{m,t}^{(i-1)}}{2 C_m} ) - L ] \nonumber\\
&\quad\quad\quad+  \frac{\alpha}{2}\sum_{t\in\calt_i}\sum_{m\in\calm}\bar{y}_{i,m}^t L,\\
&= \frac{\alpha}{2} \sum_{t\in\calt_i}\sum_{m\in\calm}\bar{y}_{i,m}^t\phi_{m,t}(w_{m,t}^{(i-1)}) +  \frac{\alpha}{2} \sum_{t\in\calt_i}\sum_{m\in\calm}\bar{y}_{i,m}^t L. \nonumber
\end{align*}
\end{subequations}

The inequality~\eqref{eq:ineq2} holds because: 
\begin{enumerate}[label=(\roman*)]
    \item $v_i \ge \frac{1}{\sigma}\sum_{t\in\calt_i}\sum_{m\in\calm} \bar{y}_{i,m}^t \phi(w^{(i-1)}_{m,t})$ from the decision rule in the online algorithm, and
    \item $\sum_{t\in\calt_i}\sum_{m\in\calm}  \bar{y}_{i,m}^t L \le \sum_{t\in\calt_i} \frac{v_i}{T_i} \le v_i$ from condition~\eqref{eq:ass-1} in Assumption~\ref{ass1}.
\end{enumerate}

Thus, we have
\begin{align*}
    \texttt{Dual}(\cali) = D_S &= \sum\nolimits_{i\in\cals} [D_i - D_{i-1}] \\
    &\le \sum\nolimits_{i\in\cals} \frac{(\sigma+1)\alpha}{2} [P_i - P_{i-1}] \\
    &= \frac{(\sigma+1)\alpha}{2} P_S = \frac{(\sigma+1)\alpha}{2} \alg(\cali),
\end{align*}
which completes the proof.

\smallskip
\noindent\textbf{EXP-Adaptive.}
Although \expp provides performance guarantees under worst-case slice arrivals, its pricing strategy can be overly conservative because it is designed to maintain feasible admission and resource allocation decisions in extreme scenarios that rarely occur in practice.
In practical deployments, statistical information about SR arrivals and resource demands may be available from historical data or prior observations. This creates an opportunity to improve \expp using a data-driven approach.

Motivated by the exponential pricing form of \expp, we learn an exponential pricing function from empirical data. Specifically, we assume the pricing function has a parametric form $\phi_{m,t;L\alpha}$ defined in ~\eqref{eq:threshold}, where $L$ and $\alpha$ are learnable parameters. In EXP-Adaptive, we select $(L,\alpha)$ from a finite candidate set $\mathcal{L} \times \mathcal{A}$, where each pair defines a distinct pricing rule. We evaluate each candidate parameter pair $(L,\alpha)\in\mathcal L\times\mathcal A$ over $N$ simulated instances and select the pair that maximizes the cumulative admitted value:
$[
(\hat L,\hat\alpha)
\in
\arg\max_{(L,\alpha)\in\mathcal L\times\mathcal A}
\sum_{n=1}^{N}u_n(L,\alpha)
]$.

Importantly, EXP-Adaptive does \emph{not} assume access to a \mbox{$\sigma$-approximate} solution for CSP. It treats the underlying resource allocation algorithm as a black box and adapts pricing to its empirical behavior.

\section{Evaluation}\label{sec:results}

\subsection{Testbed} \label{sec:testbed}

We deploy a Kubernetes-native, end-to-end 5G testbed that integrates an OpenAirInterface (OAI) RAN with an Open5GS core  \cite{oai_ran,open5gs}. The gNB is disaggregated at the F1 interface, with the Distributed Unit (DU) and Centralized Unit (CU) instantiated as separate containers. OAI NR UEs connect to the gNB through RFSIMULATOR over an additive white Gaussian noise (AWGN) channel. The testbed runs on an AMD EPYC~9254 server with 44 allocatable vCPUs and 32 GB of RAM. To minimize scheduling interference and ensure reproducible resource measurements, the Kubernetes static CPU manager pins VNFs $f_0-f_3$ to dedicated CPU cores. The OAI cell operates in band~n78 using 106 physical resource blocks (PRBs), a 30 kHz subcarrier spacing, and time-division duplexing. We replay representative eMBB, URLLC, and mMTC traffic traces from the Kestrel dataset \cite{kestrel}. During each experiment, we collect CPU, memory and PRB utilization, modulation and coding scheme (MCS), and traffic at the F1 and N3 interfaces. 

\begin{figure}
    \centering
    \includegraphics[width=0.8\linewidth,trim=0cm 0cm 0cm 0.0cm, clip]{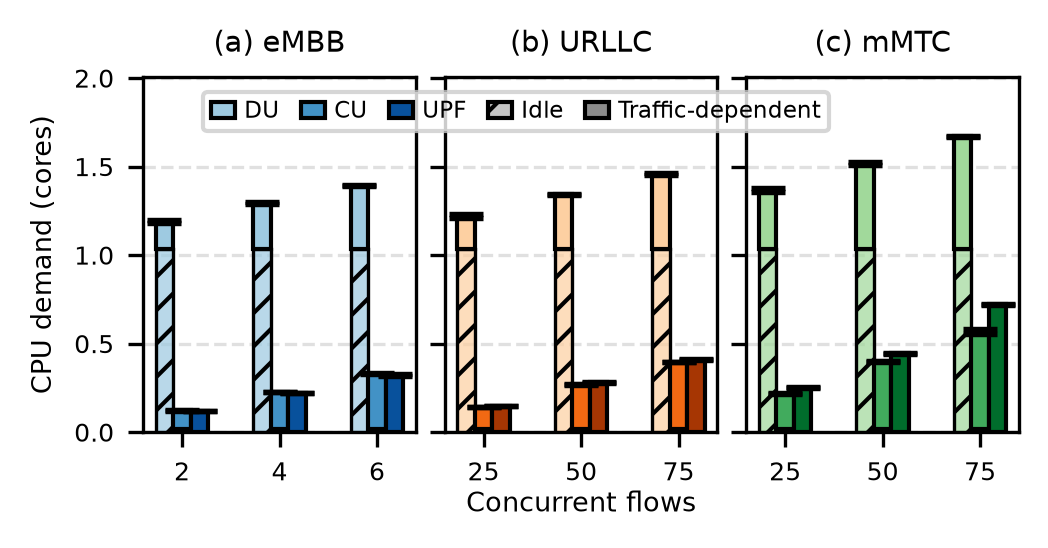}
    \caption{VNF resource usage profile using Kestrel \cite{kestrel} dataset.}
    \label{fig:profiling}
\end{figure}

\begin{figure}
    \centering
    \includegraphics[width=0.85\linewidth,trim=0.75cm 0.35cm 0cm 0.25cm, clip]{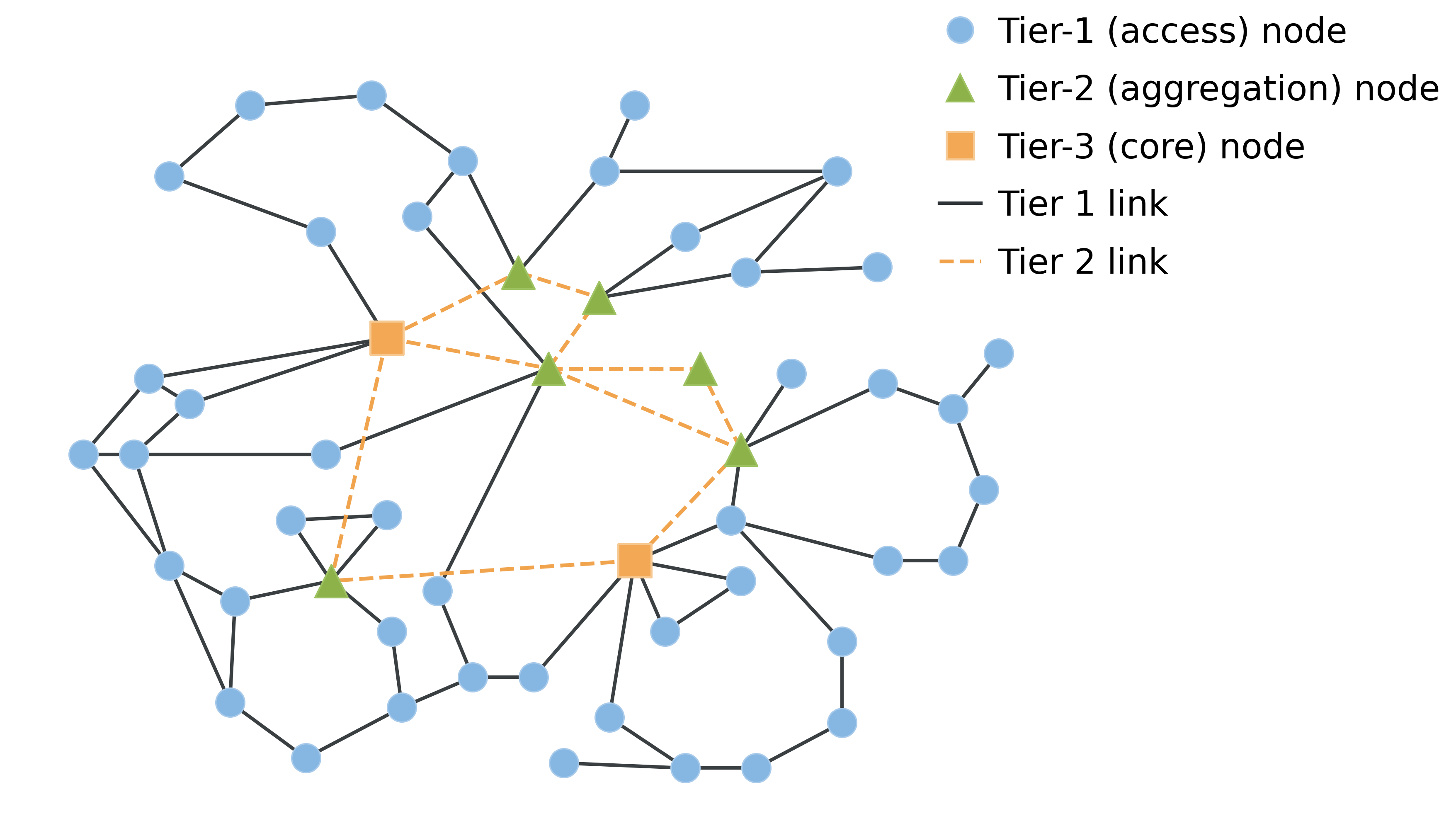}
    \caption{Telecom Italia metro-regional network \cite{askari2019latency}.}
    \label{fig:topos}
\end{figure}

Figure~\ref{fig:profiling} shows that VNF compute utilization increases linearly with the number of concurrent flows ($k$) and varies substantially across traffic classes, consistent with \cite{Slicepilot}. At the highest operating points, eMBB delivers 159.1\,Mbps and consumes 2.05 cumulative CPU cores, while URLLC delivers 177.6\,Mbps and consumes 2.26 cores. In contrast, mMTC consumes 2.96 cores despite delivering only 25.1\,Mbps. This difference is explained by its substantially higher packet rate. These measurements motivate class- and VNF-specific compute models rather than a fixed CPU requirement per slice.

\subsection{Simulation Setup}  \label{sec:sim_setup}

\begin{figure*}[ht!]
    \centering
    
    \begin{subfigure}[t]{0.3\textwidth}
        \centering
        \includegraphics[width=\linewidth]{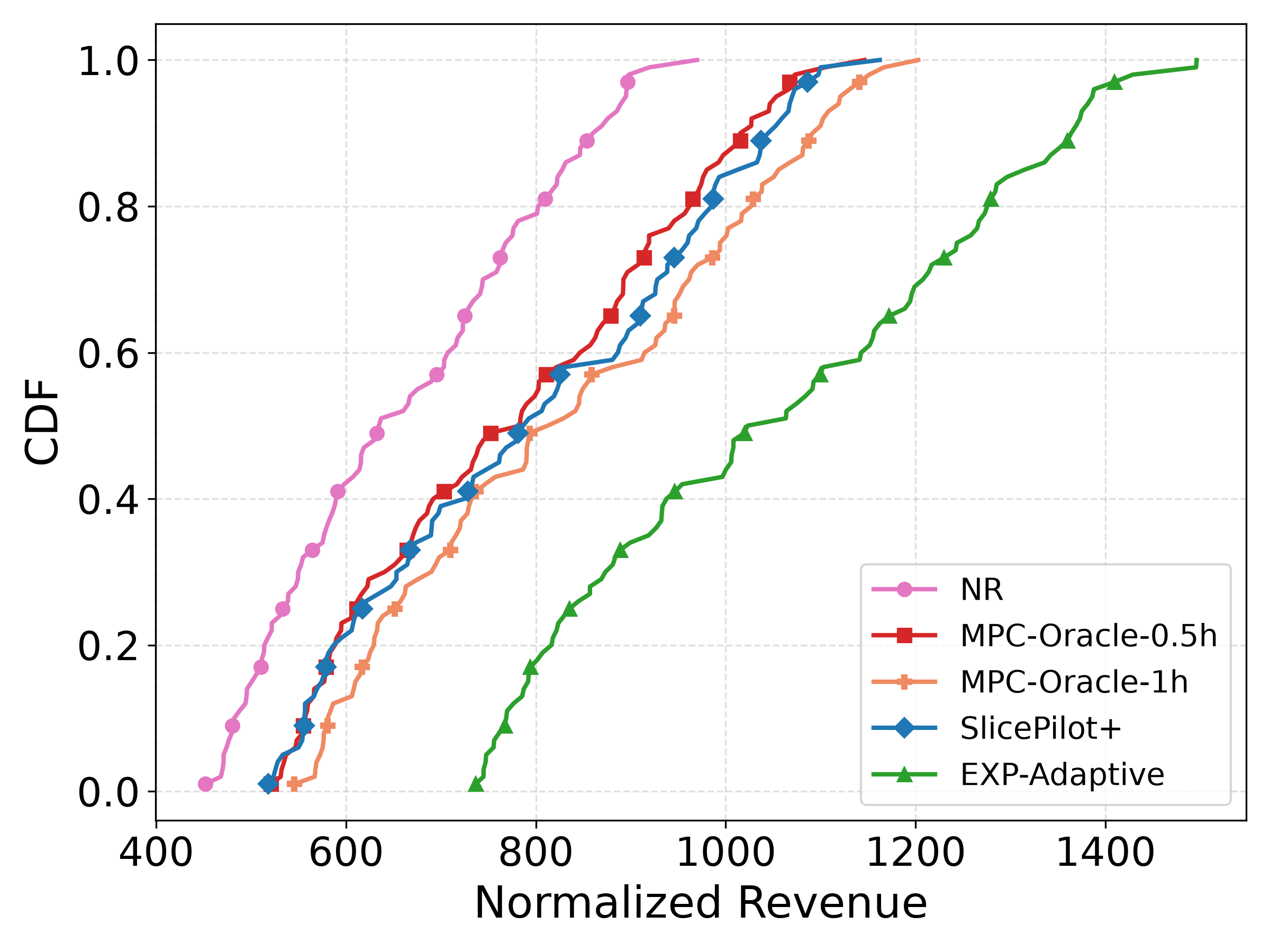}
        \caption{Total admitted norm. revenue}
        \label{fig:revenue}
    \end{subfigure}
    \hfill
    \begin{subfigure}[t]{0.3\textwidth}
        \centering
        \includegraphics[width=\linewidth]{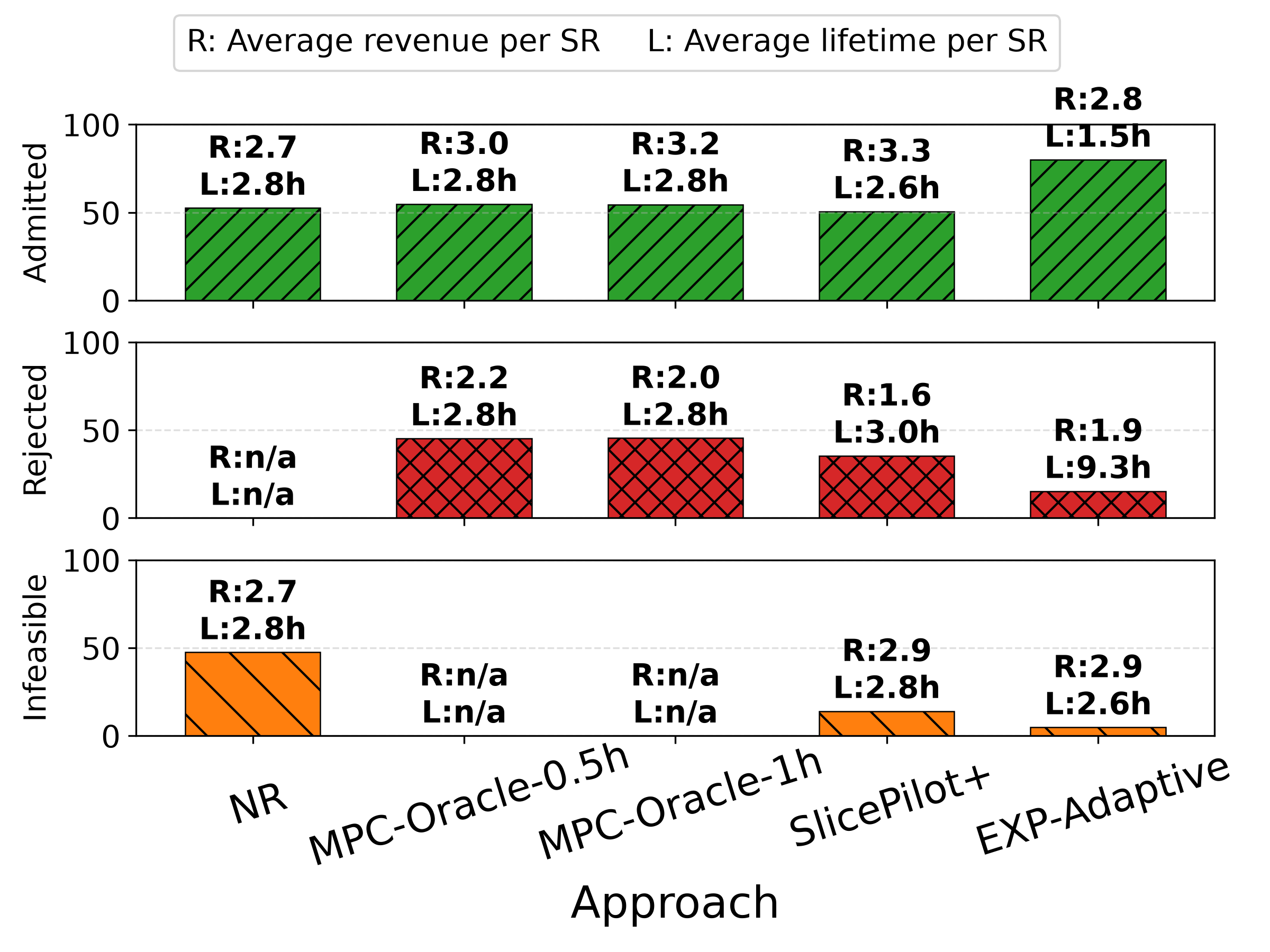}
        \caption{Admission outcomes}
        \label{fig:partial_info_admission_outcomes}
    \end{subfigure}
    \hfill
    \begin{subfigure}[t]{0.3\textwidth}
        \centering
        \includegraphics[width=\linewidth]{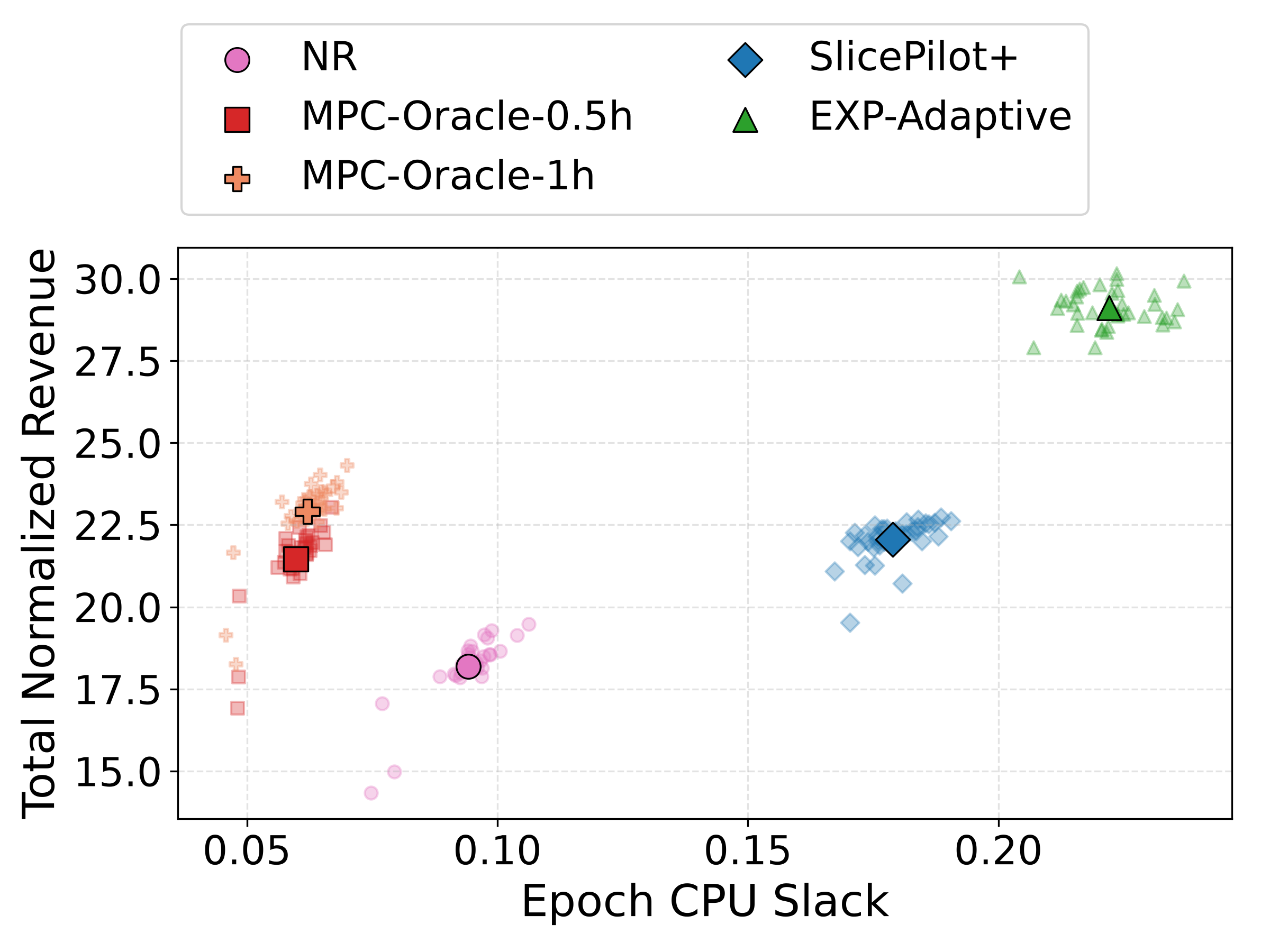}
        \caption{Revenue-capacity tradeoff}
        \label{fig:slack}
    \end{subfigure}

    \caption{Admission control and resource allocation statistics.}
    \label{fig:partial_info_stats}
\end{figure*}

We simulate OSARA for 48 hours over 100 runs, excluding the first 12 hours from the results to remove transient effects caused by the initially empty network. The following simulation parameters are used:\\
\noindent\textbf{RAN Configuration.} Each access node models a three-sector NR macro site, where each sector operates a 100 MHz $8\times8$ MIMO cell with 273 PRBs. Although PRB constraints are enforced, they remain non-binding under the evaluated workload and substrate configuration, consistent with prior work~\cite{Slicepilot,TNSM_GNN,TNET-1}. PRB-limited settings can be supported without modifying the proposed approach.\\
\noindent\textbf{Slice Requests.} We consider eMBB, URLLC, and mMTC slices, denoted by $u\in\mathcal{U}=\{1,2,3\}$, with $k\in\{10,25,50\}$, $\{50,100\}$, and $\{50,100\}$ concurrent flows, respectively. Unlike prior work which uses fixed resource requirements~\cite{Slicepilot,TNSM_GNN,TNET-1,orange_paper}, we derive VNF resource demands from testbed measurements. eMBB requirements beyond our testbed range are estimated using fitted models. The MEC applications are Varnish, LiveKit SFU, and EMQX Neuron, representing content caching, real-time video conferencing, and industrial IoT gateway, respectively. Their resource reservations follow official deployment charts ~\cite{varnish_helm_values,
varnish_intel_benchmark,varnish_cache_sizing}. The VNF delay budgets are $[0.25,2,6,10,20]$, $[0.25,2,2,2,10]$, and $[0.25,2,6,30,100]$\,ms for eMBB, URLLC and mMTC, respectively, based on ~\cite{desset2012flexible, wang2017centralize}.

\noindent\textbf{SR Traffic Model.} SR arrivals and lifetimes are derived from Microsoft Azure VM traces~\cite{azure_traces}. At the start of each run, every access node is assigned a target rate sampled from $U[1,3]$ SRs/hour, creating sufficient contention for admission and placement decisions to affect performance~\cite{orange_paper}. We then select an independent 48-hour trace window per node from the Azure VM traces, preserve relative event timing, and subsample arrivals to match the target rate. Lifetimes are clipped to $[1,12]$ hours, providing a mean of 2.78\,h with approximately $10\%$ lasting 12\,h. SR values follow a truncated Zipf distribution over $[1,10]$, with shape parameter sampled once per run from $U[0.1,2]$, providing a more flexible alternative to the binomial models used in prior work~\cite{noms22,TNSM_GNN}. Slice types are generated with equal probabilities.\\
\noindent\textbf{Substrate Network.} For the substrate network, we consider the Telecom Italia metro-regional network with 52 nodes \cite{askari2019latency}, shown in \fig{fig:topos}. The network consists of 44 access nodes, 6 aggregation nodes, and 2 core nodes. The node resource capacities are set using the model in \cite{Slicepilot}. Tier~1 and Tier~2 links are assigned transmission delays of 1.8 ms and 4.8 ms, and bandwidth capacities of 10 Gbps and 100 Gbps, respectively~\cite{NGMN_topo}.

\subsection{Comparative Approaches}
As discussed in Section~\ref{sec:bg_related}, data-driven approaches to \osac generally fall into two broad paradigms: Deep Reinforcement Learning and Model Predictive Control.

\noindent \textbf{SlicePilot+ (MARL).}\label{sec:marl_nr} We implement the MARL approach from~\cite{Slicepilot}, with a separate agent for each slice type, and extend it with explicit admission control. For each arriving SR, the corresponding agent sequentially places the VNFs, while virtual links are routed along feasible shortest-delay paths. As suggested in \cite{TNSM_GNN} for large topologies, a GATv2~\cite{gatv2} encoder processes the substrate graph, and its embeddings are combined with the slice and partial-placement features to output one Q-value per substrate node and an additional admission action. Infeasible actions are excluded through action masking. The reward combines normalized admitted revenue with penalties for infeasible placements.

\noindent \textbf{Model Predictive Control (\mpc-Oracle).} Several works address slice admission control and resource allocation through a yield-driven overbooking framework~\cite{vs1, vs2, kansaas}. To capture this, we implement a model predictive control policy, denoted \textit{MPC-Oracle}, which partitions time into epochs. At the start of each epoch, the controller constructs a forecast of slice requests expected to arrive over a finite prediction horizon and solves a revenue-maximization problem subject to the residual node and link capacities at that time. The decisions for the current epoch are implemented, after which the horizon recedes and the optimization is repeated at the next epoch with updated system state. We denote the resulting baselines according to the optimization-epoch duration, \ie MPC-Oracle-0.5h and MPC-Oracle-1h, corresponding to optimization epochs of length 0.5h and 1h, respectively. To isolate the effect of receding-horizon optimization from forecasting errors, we assume perfect knowledge of future slice arrivals.

\noindent \textbf{Node-Ranking (\nr).} \label{sec:node-ranking}
A greedy heuristic for isolation-aware RAN slicing with delay constraints~\cite{yu2020isolation}. \nr sequentially embeds VNFs by selecting the highest-ranked feasible substrate node based on a weighted residual-resource and delay score.

\subsection{Performance Comparison}

\noindent\textbf{Revenue.}
Figure~\ref{fig:revenue} reports the revenue achieved by the different
approaches. With the greedy heuristic \nr\ as a baseline, SlicePilot+,
\mpc-Oracle-0.5h, \mpc-Oracle-1h, and \expp-Adaptive increase mean revenue by
$20.6\%$, $17.9\%$, $25.9\%$, and $59.4\%$, respectively. Compared with the
state-of-the-art baselines, SlicePilot+ and \mpc-Oracle-1h, \expp-Adaptive consistently
achieves higher revenue across the observed distribution, with mean
improvements of $32.2\%$ and $26.7\%$, respectively.

\noindent\textbf{Admission Control.} For a deeper understanding of the different approaches, we analyze their
admission outcomes in Fig.~\ref{fig:partial_info_admission_outcomes}. The
greedy \nr\ policy performs no explicit rejection, resulting in $52.59\%$
admitted and $47.41\%$ infeasible SRs. Admitted and infeasible requests have
nearly identical average revenues ($2.65$ and $2.66$) and lifetimes
($2.77$ and $2.78$\,h), confirming that \nr\ does not distinguish requests
by either value or duration.

Both MPC-Oracle variants achieve zero infeasible outcomes through joint batch optimization and admit approximately $55\%$ of SRs. \mpc-Oracle-1h admits higher-value SRs on average ($3.17$ versus $2.97$) and rejects lower-value SRs ($1.98$ versus $2.23$), as its longer epoch provides a larger candidate batch. However, both optimize only within the current epoch and ignore resource occupation beyond that horizon. Hence, their admitted and rejected lifetimes remain close to the overall mean of $2.78$ h, similar to \nr. SlicePilot+ exhibits the strongest value separation with admitted SRs having the highest average revenue ($3.31$), while rejected SRs having the lowest ($1.60$). However, it shows limited lifetime selectivity, reflecting a long-horizon credit-assignment problem, \ie the impact of long-lived SRs appears only later through reduced capacity for future arrivals. 

Finally, \expp-Adaptive achieves the strongest joint lifetime and revenue selectivity. It admits $79.96\%$ of SRs, favoring shorter-lived requests with higher offered revenue ($1.55$ h and $2.78$ on average), while rejecting substantially longer-lived requests with lower revenue ($9.33$ h and $1.93$). This shows that adaptive pricing captures both immediate value and the opportunity cost of prolonged resource occupation. The remaining infeasible SRs are relatively valuable ($2.92$ on average), suggesting that they are rejected by capacity limitations rather than by deliberate deprioritization. This selectivity preserves resources for future arrivals and yields the highest aggregate revenue.

\noindent\textbf{Intertemporal Revenue-Capacity Trade-off.} Figure~\ref{fig:slack} relates CPU slack at the beginning of each epoch to the normalized revenue earned during that epoch. CPU slack denotes the fraction of aggregate substrate CPU capacity that remains unallocated at epoch start. We exclude the first 10 epochs to remove initialization transient effects. As expected, both \nr\ and the MPC-Oracle variants maintain little slack due to their myopic per-request and  per-epoch optimization, respectively. However, low slack does not necessarily yield high long-term revenue, as aggressively consuming capacity can reduce flexibility for future high-value arrivals. \nr\ achieves the lowest normalized epoch revenue, reflecting inefficient, value-agnostic resource consumption. MPC-Oracle packs requests more effectively through joint optimization and therefore earns higher revenue. The 1-hour variant benefits from a larger candidate batch and improves selection within the epoch, but it still ignores opportunity costs beyond that horizon. Consequently, both MPC variants operate near saturation and preserve little capacity for subsequent epochs.

SlicePilot+ and \expp-Adaptive maintain substantially more slack, preserving capacity for valuable future SRs. However, SlicePilot+ does not consistently convert this headroom into additional revenue. Consequently, it achieves normalized epoch revenue comparable to the MPC-Oracle variants. In contrast, \expp-Adaptive attains both the highest slack and the highest normalized epoch revenue, showing that utilization-dependent pricing preserves capacity when needed, effectively utilizing it when profitable SRs arrive.

\begin{figure}[t]
    \centering
    \includegraphics[width=0.7\linewidth]
    {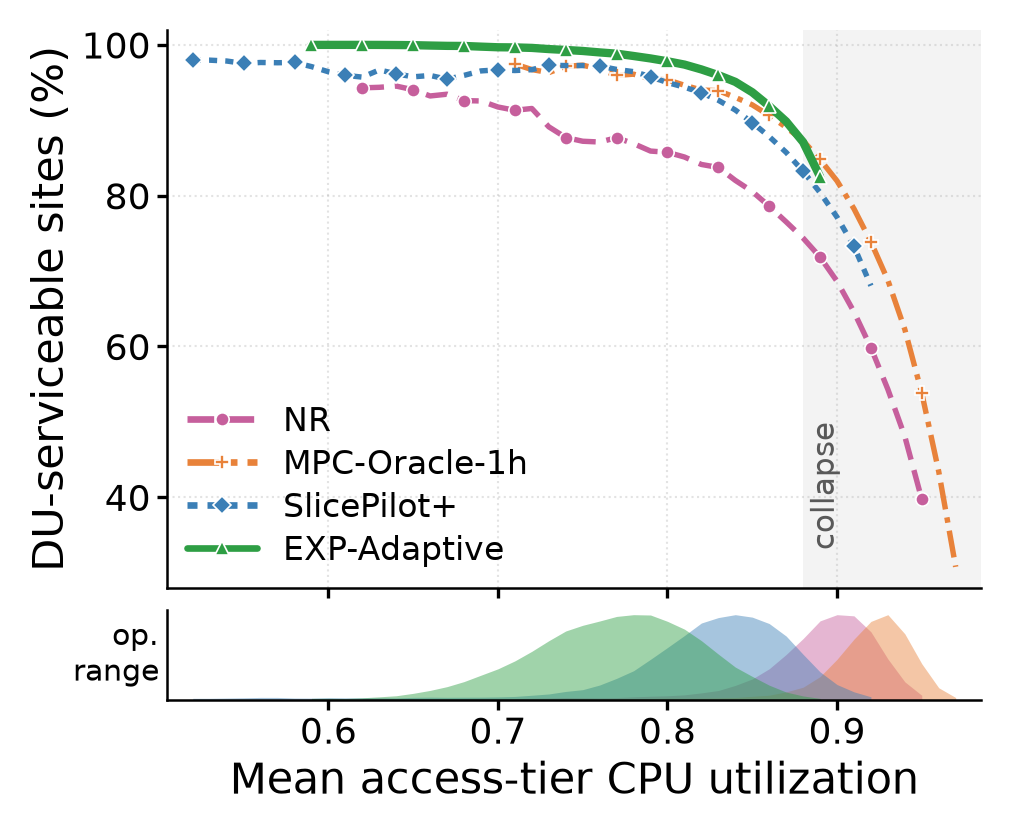}
    \caption{DU-serviceability vs. mean access-tier CPU utilization.}
    \label{fig:frag}
\end{figure}

\noindent\textbf{Resource Allocation.} In \osac, where admission control and resource allocation are jointly optimized, aggregate utilization alone does not indicate resource allocation quality. A method may appear underutilized because it selectively preserves capacity for future high-value SRs. Figure~\ref{fig:frag} therefore relates DU serviceability to mean access-tier CPU utilization, while the lower panel shows each method’s empirical operating range. We focus on the DU because it is the first placeable VNF after the fixed RU and typically faces the tightest latency constraint. An access site is DU-serviceable if a standardized DU can be placed locally or at a delay-reachable node while satisfying residual resource constraints.

The greedy \nr\ policy loses serviceability earliest, indicating that its local decisions fragment residual compute and bandwidth. In contrast, the state-of-the-art baselines, SlicePilot+ and \mpc-Oracle-1h, together with \expp-Adaptive, preserve substantially more DU-serviceable sites, demonstrating the benefit of strategic resource allocation. At very high utilization, \mpc-Oracle-1h preserves the most serviceable sites due to its look-ahead optimization. SlicePilot+ remains close to MPC despite operating fully online. Among the online methods, \expp-Adaptive performs best and largely avoids the collapse region by combining effective placement with admission pricing, thereby preserving feasible embeddings for future SRs.

\noindent\textbf{Interpretability.} Attention-based GNNs (GATv2) can explain allocation through attention scores, but admission decisions remain black-boxed~\cite{TNSM_GNN}. In contrast, \opa bases both on resource prices, which \expp-Adaptive makes utilization-dependent. Figure~\ref{fig:prices} shows that Node (cpu, mem), and link (bandwidth) resource prices increase smoothly with pre-decision utilization, providing interpretable signals of scarcity and opportunity cost. These prices guide the minimum-cost feasible placement, after which its cost is compared with the slice revenue to decide if the SR should be admitted. Each admission decision can therefore be traced from utilization to prices, placement cost, and admission outcome.

\begin{figure}[t]
    \centering
    \begin{subfigure}[t]{0.49\columnwidth}
        \centering
        \includegraphics[width=\linewidth]
        {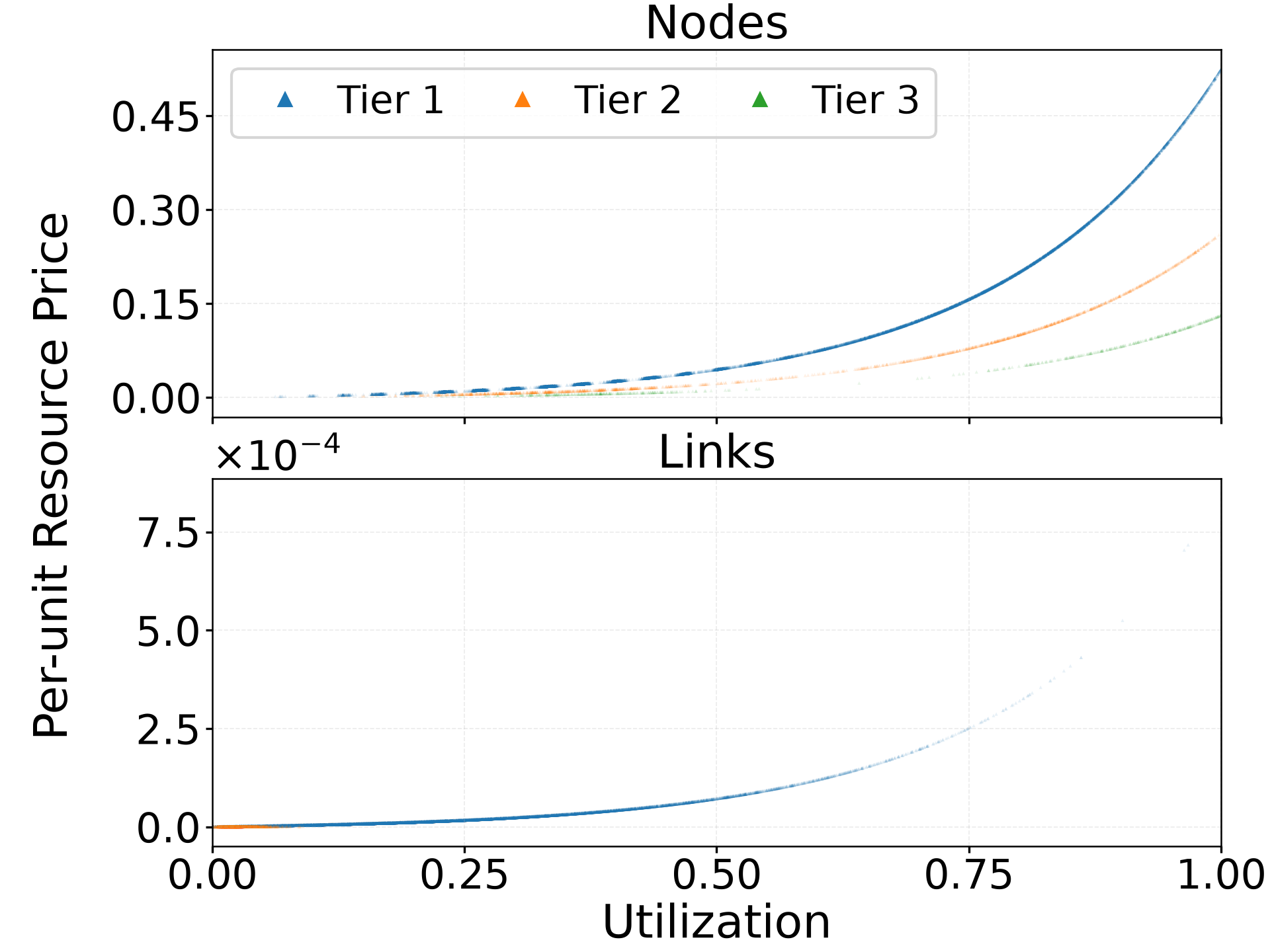}
        \caption{Empirical resource prices}
        \label{fig:prices}
    \end{subfigure}
    \hfill
    \begin{subfigure}[t]{0.49\columnwidth}
        \centering
        \includegraphics[width=\linewidth]
        {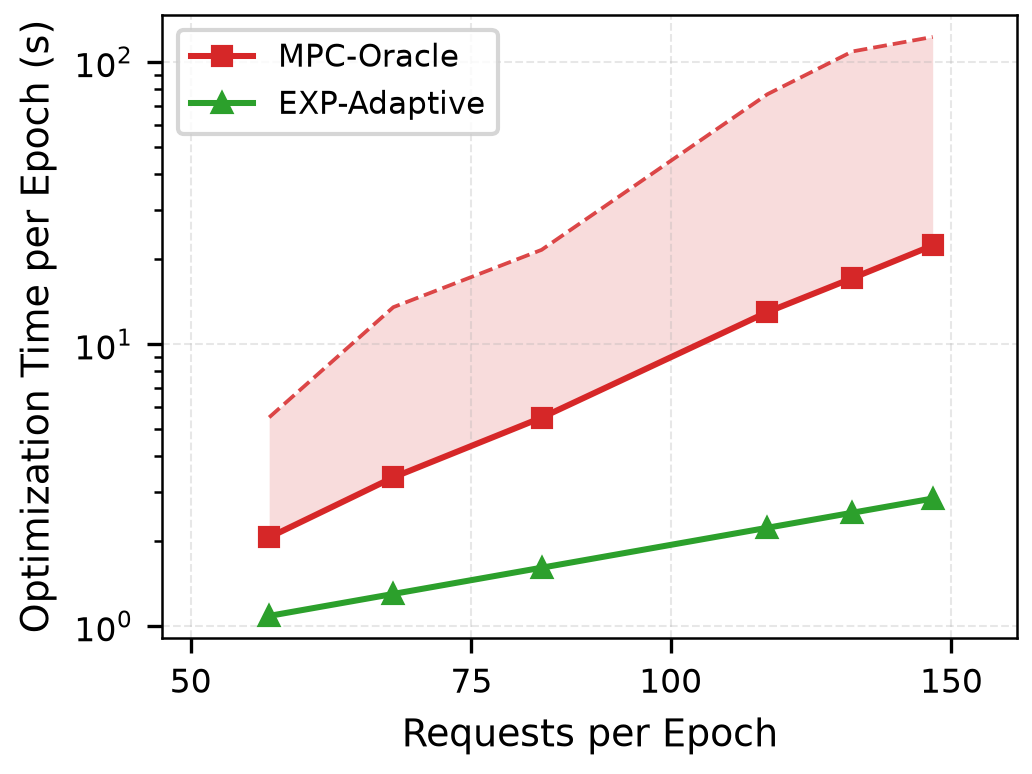}
        \caption{Optimization time}
        \label{fig:batch_size_vs_solve_time}
    \end{subfigure}
    \caption{Interpretability and computational complexity plots.}
    \label{fig:interpretability_complexity}
\end{figure}

\noindent\textbf{Computational Complexity.} For DRL approaches, online cost is dominated by neural-network inference, which grows with substrate graph size. Despite substantial training overhead, inference cost is typically comparable to heuristic methods~\cite{gatv2,TNSM_GNN}. In \opa, each SR independently solves a \csp whose size depends only on its candidate nodes and links, so total decision time grows approximately linearly with the number of arrivals. In contrast, MPC jointly optimizes all SRs within an epoch, with overlapping lifetimes coupling them through shared capacity constraints. Figure~\ref{fig:batch_size_vs_solve_time} compares optimization time across evaluated runs. MPC-Oracle’s median runtime rises from 2.07s to 22.38s, while its 95th percentile grows from 5.50s to 122.79s as the epoch size increases from 56 to 146 SRs. Over the same range, \expp-Adaptive's runtime increases from 0.96s to 2.40s at the median and from 1.12s to 2.69s at the 95th percentile, yielding $9.3\times$ and $45.7\times$ lower runtimes, respectively, at the upper end. This comparison favors MPC-Oracle, which assumes perfect forecasts and solves only once per epoch. A fully online MPC implementation would require repeated re-optimization after each arrival.
\section{Conclusion} \label{sec:conclusion}
We presented a pricing-based framework for online slice admission control and resource allocation in 5G and beyond mobile networks. The proposed approach processes SRs sequentially and determines admission and resource allocation by solving a single-slice optimization problem. These decisions are guided by dynamically updated resource prices that capture resource scarcity and opportunity cost. This design enables the framework to balance immediate admissions with preserving capacity for future requests while maintaining tractable online operation. We further developed an exponential pricing strategy that provides bounded worst-case performance guarantees, and we built on this design to develop an adaptive variant that learns its pricing parameters from data to improve empirical performance. We compared the proposed approach against state-of-the-art methods from the DRL- and MPC-based literature and show that it improves mean revenue by up to $32.2\%$ over DRL and $26.7\%$ over an MPC baseline with perfect forecasts while having lower computational cost.

For future work, we plan to investigate online learning approaches, study dynamic market settings where admission decisions can influence future SR arrivals, and develop scalable approximation algorithms for \csp.
\section*{Acknowledgement}
This work was supported in part by the Ontario Research Fund – Research Excellence program (Project\# ORF-RE012-051) from the Province of Ontario. The views expressed herein are those of the authors and do not necessarily reflect those of the Province.


%
%


%
%


\bibliographystyle{IEEEtran}
\bibliography{references}

\end{document}